\documentclass[journal]{IEEEtran}

\ifCLASSINFOpdf
\else
   \usepackage[dvips]{graphicx}
\fi
\usepackage{url}
\usepackage{mathtools}
\usepackage{graphicx,amsfonts,amsmath,amsthm,color}
\newtheorem{theorem}{Theorem}
\newtheorem{lemma}{Lemma}
\newcommand{\rF}{\rm F}

\newcommand{\argmin}{\mathop{\mathrm{argmin}}}
\newcommand{\argmax}{\mathop{\mathrm{argmax}}}

\begin{document}

\title{Bounded Projection Matrix Approximation with Applications to Community Detection}

\author{Zheng Zhai, Hengchao Chen, and Qiang Sun

\thanks{Manuscipt received \today.}
\thanks{Zheng Zhai is with the College of Artificial Intelligence, Dalian Maritime University, and the Department of
Statistical Sciences, University of Toronto. Hengchao Chen and Qiang Sun are with the Department of
Statistical Sciences, University of Toronto.}
\thanks{Corresponding author: Qiang Sun. E-mail: \texttt{qiang.sun@utoronto.ca}.}
}

\markboth{Journal of \LaTeX\ Class Files, Vol. 14, No. 8, August 2015}
{Shell \MakeLowercase{\textit{et al.}}: Bare Demo of IEEEtran.cls for IEEE Journals}
\maketitle

\begin{abstract}

Community detection is an important problem in unsupervised learning. This paper proposes to solve a projection matrix approximation problem with an additional entrywise bounded constraint. Algorithmically, we introduce a new differentiable convex penalty and derive an alternating direction method of multipliers (ADMM) algorithm. Theoretically, we establish the convergence properties of the proposed algorithm. Numerical experiments demonstrate the superiority of our algorithm over its competitors, such as the semi-definite relaxation method and spectral clustering. 

\end{abstract}


\IEEEpeerreviewmaketitle

\section{Introduction}

\IEEEPARstart{C}{ommunity} detection is an important problem in unsupervised learning that has attracted the attention of researchers from various fields, such as mathematics, statistics, applied mathematics, physics, and social sciences. The goal of this problem is to partition $n$ data points into $K$ groups based on their pairwise similarities, which can be represented as a similarity matrix $A\in\mathbb{R}^{n\times n}$. A common approach to solve this problem is to first derive a lower-dimensional representation \cite{belkin2003laplacian,cai2018comprehensive,goyal2018graph,xu2021understanding} of the data from $A$ and then apply a clustering algorithm such as  $k$-means \cite{hartigan1979k} or EM algorithm \cite{jung2014clustering} to identify the clusters. The efficacy of this method is contingent on the quality of the data representation. 

A popular choice for the data representation is to use the top $K$ eigenvectors of $A$ as in spectral clustering~\cite{von2007tutorial}. Finding these eigenvectors is equivalent,  up to rotations, to determining the subspace spanned by these vectors. The latter is also equivalent to the following projection matrix approximation problem:
\begin{align}
X=\argmin_{X\in\mathcal{P}_{K}}\|A-X\|^2_{\rF},\label{equ:1}
\end{align}
where $\mathcal{P}_{K}\subseteq\mathbb{R}^{n\times n}$ is the set of rank-$K$ projection matrices. Thus, the effectiveness of spectral clustering is highly dependent on the quality of the projection matrix approximation. 

The unconstrained projection matrix approximation~\eqref{equ:1} may be less effective when extra information is available. 
In community detection, for instance, an intermediate step is to estimate the projection matrix $X$ associated with the assignment matrix $\Theta\in\{0,1\}^{n\times K}$, where $\Theta_{ik}=1$ if and only if node $i$ belongs to group $k$; see Section \ref{sec:cd} for details. Such a  projection matrix has certain structures:
\begin{itemize}
    \item[1)] $X$ has non-negative elements.
    \item[2)] $X$ has elements upper bounded by $\max_k\frac{1}{n_k}$, where $n_k$ is the size of the $k$-th group.
\end{itemize}
Hence, it may be beneficial to seek projection matrices with these desired structures enforced. Inspired by this example, we propose to study the following bounded projection matrix approximation (BPMA) problem:
\begin{align}
X=\argmin_{X\in {\cal P}_K \atop X_{ij}\in [\alpha,\beta] } \|A- X\|_{\rm F}^2,\label{BPA}
\end{align}
where $\alpha,\beta\in\mathbb{R}$ are lower and upper bounds set {\it a priori}. In the above example, we simply set $\alpha=0$ and $\beta=\max_k\frac{1}{n_k}$. 

Due to the projection matrix and boundedness constraints, it is challenging to solve \eqref{BPA} directly. To address this difficulty, this paper proposes a new differentiable convex penalty to relax the boundedness constraint. We employ the alternating direction method of multipliers (ADMM) \cite{boyd2011distributed,sun2022efficient} to solve this relaxed problem. Moreover, we show that any limiting point of the solution sequence is a stationary point of the relaxed problem. Finally, we apply the proposed method to community detection and demonstrate its superiority over its competitors in both synthetic and real world datasets.

\subsection{Related Work}

Low-rank matrix optimization with additional structural constraints is a common problem in machine learning and signal processing \cite{zhang2012low,zhang2012inducible}. The problem aims to find the best low-rank matrix approximation that also satisfies certain structural constraints, such as non-negativity, symmetry, boundedness, and sparsity. One line of research studies the matrix factorization approach, such as non-negative matrix factorization~\cite{lee1999learning}, semi-nonnegative matrix factorization~\cite{ding2008convex}, bounded low-rank matrix approximation~\cite{kannan2012bounded}. Another line of research studies simultaneously low-rank and sparse matrix approximation~\cite{zhang2022graph,ji2011robust,richard2012estimation}. These works, however, only seek a low-rank matrix, which is not necessarily a projection matrix. In contrast, motivated by the problem of community detection, our paper studies the projection matrix approximation problem with additional boundedness constraints. We then propose an ADMM algorithm, and prove the convergence properties.

\section{Bounded Projection Matrix Approximation}

In this section, we study how to solve \eqref{BPA}. First, we relax the BPMA problem \eqref{BPA} using a differentiable convex penalty. Then we derive the ADMM algorithm that can efficiently solve the relaxed BPMA problem. 

\subsection{Differentiable Convex Penalty}

To start with, we define the following indicator function:
\[
I_{\alpha,\beta}(x) = \left\{ 
\begin{array}{ll}
+\infty, & x < \alpha\\
0, &\alpha \leq x \leq \beta \\
+\infty, & x > \beta.
\end{array} \right.
\] 
Then we can rewrite the BPMA problem \eqref{BPA} as the following optimization problem:
\begin{equation}\label{indictor_opt}
X=\argmin_{X\in {\cal P}_K} \|A-X\|_{\rm F}^2+ \sum_{ij} I_{\alpha,\beta}(X_{ij}).
\end{equation}
It is challenging to solve problem \eqref{indictor_opt} due to the discontinuity of the penalty function $I_{\alpha,\beta}(\cdot)$. To alleviate this issue, we propose  to replace $I_{\alpha,\beta}(\cdot)$ by $\lambda g_{\alpha,\beta}(\cdot)$ and solve
\begin{align}
    \label{Convex} 
    X=\argmin_{X\in\mathcal{P}_K}\|A-X\|_{\rF}^2+\lambda\sum_{ij}g_{\alpha,\beta}(X_{ij}),
\end{align}
where $\lambda>0$ is a tuning parameter and $g_{\alpha,\beta}(\cdot)$ is a differentiable convex penalty function given by
\begin{equation}\label{gab}
g_{\alpha,\beta}(x) =  (\min\{x-\alpha, 0\})^2 + (\min\{\beta-x,0\})^2.
\end{equation}
Since $g_{\alpha,\beta}(\cdot)$ is non-negative and $g_{\alpha,\beta}(x)=0$ if and only if $x\in[\alpha,\beta]$, problem \eqref{Convex} reduces to problem \eqref{indictor_opt} when $\lambda\to\infty$. We shall refer to problem \eqref{Convex} as the relaxed bounded projection matrix approximation (RBPMA) problem.

\subsection{Algorithm}

In this subsection, we develop an ADMM algorithm to  solve the RBPMA problem. First, we define the augmented Lagrangian ${\cal L}_\rho (X, Y, \Lambda)$ as
\[ 
\begin{aligned}
{\cal L}_\rho (X, Y, \Lambda) =&  \|A-X\|_{\rm F}^2 + \lambda\sum_{ij} g_{\alpha,\beta}(Y_{ij}) + \frac{\rho}{2} \|X-Y\|_{\rm F}^2  \\
&+\langle \Lambda, X-Y \rangle,\qquad \forall X,Y,\Lambda\in\mathbb{R}^{n\times n}.
 \end{aligned}
\]
Starting from initialization points $\{X^0,Y^0,\Lambda^0\}$, our algorithm updates $\{X^k,Y^k,\Lambda^k\}$ alternatively as:
\begin{align}
X^{k+1} &=\argmin_{X\in {\cal P}_K} {\cal L}_\rho (X, Y^k, \Lambda^k),\label{X}\\
Y^{k+1} &=\argmin_{Y} {\cal L}_\rho (X^{k+1}, Y, \Lambda^k),  \label{Y}\\
\label{lambda}
\Lambda^{k+1} &= \Lambda^k + \rho(X^{k+1}-Y^{k+1}) .
\end{align}
These updates have closed-form solutions and thus can be implemented efficiently. Specifically, problem \eqref{X} is equivalent to the following problem:
\[
X^{k+1}=\argmax_{X\in {\cal P}_K} \langle X,W^k\rangle,\quad W^k=2A+\rho Y^k-\Lambda^k,
\]
and $X^{k+1}$ is given by the projection matrix associated with the leading $K$ eigenvectors of $W^{k}$. Problem \eqref{Y} is equivalent to: 
\[
Y^{k+1}=\argmin_{Y}\|Y-V^{k+1}\|_{\rF}^2+\tau\sum_{ij}g_{\alpha,\beta}(Y_{ij}),
\]
where $V^{k+1}=X^{k+1}+\Lambda^k/\rho$ and $\tau={2\lambda}/{\rho}$. This is a separable problem, and each entry $Y^{k+1}_{ij}$ can be solved by
\begin{align*}
Y^{k+1}_{ij}=\argmin_{Y_{ij}}\left(Y_{ij}-V^{k+1}_{ij}\right)^2+\tau g_{\alpha,\beta}(Y_{ij}).
\end{align*}
In a compact form, the solution $Y^{k+1}$ can be written as
\[
Y^{k+1}=\frac{V^{k+1}+\tau\mathcal{P}_{\alpha,\beta}(V^{k+1})}{1+\tau},
\]
where $\mathcal{P}_{\alpha,\beta}(\cdot)$ is an entrywise projection operator given by
\[
\mathcal{P}_{\alpha,\beta}(V)=\min\{\max\{V,\alpha\},\beta\}.
\]
Here $\min$ and $\max$ are defined entrywise.

\section{Convergence Theory}

This section provides convergence properties of the proposed ADMM algorithm. We show that any limiting point of the solution sequence is a stationary point of problem \eqref{Convex}. Our proof consists of three components. First, we show in Lemma~\ref{dual} that the successive change of the dual variable $\Lambda$ is controlled by that of $Y$.

\begin{lemma}\label{dual}
 $\|\Lambda^{k+1}-\Lambda^k\|_{\rF}\leq 2\lambda\|Y^{k+1}-Y^{k}\|_{\rF}$.
\end{lemma}
\begin{proof}
By definition \eqref{lambda} of $\Lambda^{k+1}$, we have
\begin{align}
\Lambda_{ij}^{k+1}&=\Lambda_{ij}^{k}+\rho(X^{k+1}_{ij}-Y^{k+1}_{ij})=\lambda g_{\alpha,\beta}'(Y^{k+1}_{ij}),\label{dg}
\end{align}
where the second equality uses the fact that $Y^{k+1}$ is a stationary point of $\mathcal{L}_{\rho}(X^{k+1},Y,\Lambda^k)$. By \eqref{gab}, we can compute
\[
g'_{\alpha,\beta}(x) = 2 \min\{x-\alpha,0\}-2\min\{\beta-x,0\}. 
\]
This is a Lipschitz continuous function with Lipschitz constant $2$. Hence, we have
\[
|\Lambda^{k+1}_{ij}-\Lambda^{k}_{ij}|= \lambda|g_{\alpha,\beta}'(Y^{k+1}_{ij})-g_{\alpha,\beta}'(Y_{ij}^k)|\leq 2\lambda|Y^{k+1}_{ij}-Y^{k}_{ij}|.
\]
Summing over all $i,j$, we prove the lemma.
\end{proof}



Next, we show that ${\cal L}_{\rho}(X^{k},Y^{k},\Lambda^{k})$ is decreasing in $k$ and the difference is lower bounded by $\lambda\|Y^{k+1}-Y^k\|_{\rF}^2$ when we set $\rho=4\lambda$.

\begin{lemma} \label{covergence}
Let $\rho=4\lambda$. The following inequality holds:
\begin{align*}
D&\coloneqq{\cal L}_\rho (X^{k}, Y^k, \Lambda^k)-{\cal L}_\rho (X^{k+1}, Y^{k+1}, \Lambda^{k+1})\\
&\geq \lambda \|Y^{k+1}-Y^k\|_{\rF}^2.
\end{align*}
\end{lemma}

\begin{proof}
First, we write $D=D_1+D_2$, where
\begin{align*}
D_1&\coloneqq {\cal L}_\rho (X^k, Y^k,\Lambda^k) - {\cal L}_\rho (X^{k+1}, Y^k, \Lambda^k),\\
D_2&\coloneqq
{\cal L}_\rho (X^{k+1}, Y^k, \Lambda^k) - {\cal L}_\rho (X^{k+1}, Y^{k+1},\Lambda^{k+1}).
\end{align*}
The term $D_1$ is non-negative because of \eqref{X}. 
For the term $D_2$, we have
\begin{align*}
D_2=E_1+E_2+E_3,
\end{align*}
where
\begin{align*}
E_1&=\frac{\rho}{2} \left(\|X^{k+1}-Y^k\|_F^2-  \|X^{k+1}-Y^{k+1}\|_F^2\right),\\
E_2&= \langle \Lambda^k, X^{k+1}-Y^k \rangle -  \langle \Lambda^{k+1}, X^{k+1}-Y^{k+1} \rangle, \\
E_3&=\lambda \sum_{ij}(g_{\alpha,\beta}(Y^k_{ij}) - g_{\alpha,\beta}(Y^{k+1}_{ij}) ).
\end{align*}
The term $E_1$ can be rewritten as
\begin{align*}
E_1&=\frac{\rho}{2}\|Y^{k+1}-Y^{k}\|_{\rF}^2+\rho\langle X^{k+1}-Y^{k+1},Y^{k+1}-Y^k\rangle\\
&=\frac{\rho}{2}\|Y^{k+1}-Y^{k}\|_{\rF}^2+\underbrace{\langle \Lambda^{k+1}-\Lambda^{k},Y^{k+1}-Y^k\rangle}_{E_4}.
\end{align*}
The term $E_2$ can be rewritten as
\begin{align*}
E_2
&= \langle \Lambda^k - \Lambda^{k+1}, X^{k+1}-Y^{k+1}\rangle+\langle \Lambda^k, Y^{k+1}-Y^k \rangle\\
&=-\frac{1}{\rho} \|\Lambda^{k+1}-\Lambda^{k}\|_{\rF}^2+\underbrace{\langle \Lambda^k, Y^{k+1}-Y^k \rangle}_{E_5},
\end{align*}
where the last equality uses \eqref{lambda}. Since $\Lambda^{k+1}_{ij}=\lambda g_{\alpha,\beta}'(Y^{k+1}_{ij})$ by \eqref{dg} and $g_{\alpha,\beta}(\cdot)$ is a convex function, we have
\begin{align*}
&E_3+E_4+E_5\\
=\ &
\lambda\sum_{ij}\left[g_{\alpha,\beta}(Y^{k}_{ij})-g_{\alpha,\beta}(Y^{k+1}_{ij})-g_{\alpha,\beta}'(Y_{ij}^{k+1})(Y_{ij}^k-Y_{ij}^{k+1})\right]\\
\geq\ & 0.
\end{align*}
Combining the above analysis, we have
\begin{align*}
D& \geq  \frac{\rho}{2} \|Y^{k+1}-Y^k\|_{\rF}^2 - \frac{1}{\rho} \|\Lambda^k-\Lambda^{k+1}\|_{\rF}^2\\
&\geq\lambda\|Y^{k+1}-Y^k\|_{\rF}^2,
\end{align*}
where we use Lemma \ref{dual} and $\rho=4\lambda$.
\end{proof}


Finally, we show that any limiting point of $\{X^k,Y^k,\Lambda^k\}$ is a stationary point of problem \eqref{Convex}.

\begin{theorem}
Let $\rho = 4\lambda$ and $\{X^*,Y^*,\Lambda^*\}$ be a limiting point of $\{X^k,Y^k,\Lambda^k\}$. Then the KKT conditions of problem \eqref{Convex} hold:
\begin{equation}
    \left\{
    \begin{aligned}
&X^*=Y^*,\ X^*\textnormal{ is the projection matrix associated}\\
&\quad \textnormal{ with the top }K\textnormal{ eigenvectors of }W^*,\\
&\Lambda_{ij}^*+\rho (Y^*_{ij}-X_{ij}^*)=\lambda g_{\alpha,\beta}'(Y^*_{ij}),\quad \forall i,j,
\end{aligned}\right.\label{kkt}
\end{equation}
where $W^*=2A+\rho Y^*-\Lambda^*$.
\end{theorem}

\begin{proof}
Since $\Lambda^{k+1}_{ij}=\lambda g_{\alpha,\beta}'(Y^{k+1}_{ij})$ by \eqref{dg} and  $g_{\alpha,\beta}(\cdot)$ is a convex function, we have
\[
\lambda \sum_{ij}g_{\alpha,\beta}(Y^{k+1}_{ij}) + \langle \Lambda^{k+1}, X^{k+1}-Y^{k+1}\rangle\geq \lambda \sum_{ij}g_{\alpha,\beta}(X^{k+1}_{ij}).
\]
Substituting this into ${\cal L}_\rho (X^{k+1}, Y^{k+1}, \Lambda^{k+1})$, we obtain
{ 
\begin{align}
 &{\cal L}_\rho (X^{k+1}, Y^{k+1}, \Lambda^{k+1}) \notag  \\  \geq\ & \|A-X^{k+1}\|_{\rF}^2 + \frac{\rho}{2} \|X^{k+1}-Y^{k+1}\|_{\rF}^2 + \lambda \sum_{ij} g_{\alpha,\beta}(X^{k+1}_{ij})\notag\\
  \geq \ &0\label{positive}
\end{align}}
Then by Lemma \ref{covergence},
\begin{align}
& \lambda\sum_{k=0}^s  \|Y^{k+1}-Y^k\|_{\rF}^2\notag \\
\leq\ & \sum_{k=0}^s {\cal L}_\rho (X^{k}, Y^{k}, \Lambda^{k}) - {\cal L}_\rho (X^{k+1}, Y^{k+1}, \Lambda^{k+1})\notag \\
=\ &{\cal L}_\rho (X^{0}, Y^{0}, \Lambda^{0}) - {\cal L}_\rho (X^{s+1}, Y^{s+1}, \Lambda^{s+1})\notag \\
\leq \ &{\cal L}_\rho (X^{0}, Y^{0}, \Lambda^{0}),\label{sum}
\end{align}
where the last inequality uses \eqref{positive}. \eqref{sum} implies that $\|Y^{k+1}-Y^k\|_{\rF}$ converges to zero, and by Lemma~\ref{dual},  $\|\Lambda^{k+1}-\Lambda^k\|_{\rF}\leq 2\lambda\|Y^{k+1}-Y^k\|_{\rF}$ also converges to zero. By \eqref{lambda}, $X^{k+1}-Y^{k+1}$ also converges to zero.



Suppose $\{k_u\}_{u=1}^{\infty}$ is a sequence satisfying
\[
\lim_{u\to\infty}(X^{k_u},Y^{k_u},\Lambda^{k_u})=(X^*,Y^*,\Lambda^*).
\]
The convergence of $X^{k_u}-Y^{k_u}$ to zero implies that $X^*=Y^*$. Also, $Y^{k_u-1}-Y^*$ and $\Lambda^{k_u-1}-\Lambda^*$ converge to zero because $Y^{k_u}-Y^{k_u-1}$ and $\Lambda^{k_u}-\Lambda^{k_u-1}$ converge to zero. By \eqref{X}, $X^{k_u}$ is the projection matrix associated with the top eigenvectors of $W^{k_u-1}=2A+\rho Y^{k_u-1}-\Lambda^{k_{u}-1}$. Since $W^{k_u-1}$ converges to $W^*$, $X^{k_u}$ converges to the projection matrix associated with the top $K$ eigenvectors of $W^*$. This proves the first condition in \eqref{kkt}. By \eqref{Y}, we have
\[
\Lambda^{k_u-1}_{ij}+\rho(Y_{ij}^{k_u}-X_{ij}^{k_u})=\lambda g_{\alpha,\beta}'(Y_{ij}^{k_u}),\quad \forall i,j.
\]
Letting $u$ tend to infinity, we conclude the proof of \eqref{kkt}. 
\end{proof}


\begin{figure*}[h] 
   \centering
   \includegraphics[width=0.9\linewidth]{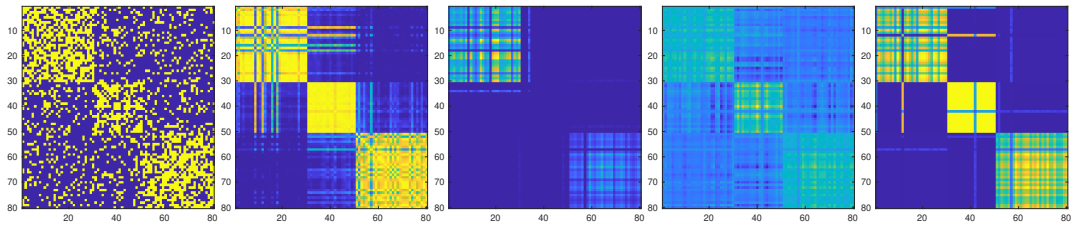} 
   \caption{The comparison between the similarity matrix $A$ and the solved projection matrix $X$ using SDP-1, SDP-2, spectral clustering, and our method. Here $A$ is generated by $A_{ij}\sim {\rm Bernoulli}(\Psi_{k\ell})$ if $i\in S_k$ and $j\in S_\ell$. We set $\Psi_{k\ell}= 0.20, \forall k\neq \ell$ and $\Psi_{kk} = 0.58,\forall k$.
   }
   \label{fig:example}
\end{figure*}
\begin{table*}[t]
\centering
\caption{The objectives and constraints for related algorithms\label{related algorithm}}
\resizebox{15cm}{!}{
\begin{tabular}{c|c|c|c|c}\hline \hline
&SDP-1&SDP-2&Spectral Method& BPMA\\ \hline
Objectives &$\max_X \ \langle A, X\rangle $ & $\max_X \ \langle A, X\rangle$ & $\max_X \ \langle A, X\rangle$ & $\max_X \ \langle A, X\rangle$ \\ \hline
Constraints & $
\begin{aligned}
&\ X {\bf 1}_n = n/K {\bf 1}_n,\ X\succeq 0, \\
&\ {\rm diag}(X) = {\bf 1}_n,\ X\geq 0 
\end{aligned}
$
& 
 $
\begin{aligned}
&\ \langle X, E_n \rangle = n^2/K,\ X\succeq 0, \\
&\ {\rm trace}(X) = n, \ 0 \leq X \leq 1 
\end{aligned}
$
& 
$
\begin{aligned}
 &\ X\in {\cal P}_K\\ 
 \end{aligned}
$
& 
$
\begin{aligned}
 &\ X\in {\cal P}_K \\
 &\ X_{i,j}\in [\alpha,\beta]
\end{aligned}
$\\ \hline \hline
\end{tabular}}
\end{table*}

\section{Community Detection}\label{sec:cd}

In this section, we apply the proposed algorithm to community detection. Specifically, we shall consider stochastic block models (SBM)~\cite{abbe2017community}. Suppose there are $n$ data points and $K$ groups. Let $S_{k}\subseteq\{1,\ldots,n\}$ be the index set of all data points in the $k$-th group. Let $\Psi\in[0,1]^{K\times K}$ be a symmetric matrix with $\Psi_{k\ell}$ representing the connectivity probability between the $k$-th group and the $\ell$-th group. The SBM assumes that the similarity matrix $A\in\{0,1\}^{n\times n}$ is a symmetric matrix. Moreover, $A_{ij}\sim{\rm Bernoulli}(\Psi_{k\ell})$ when $i\leq j$, $i\in S_{k}$ and $j\in S_{\ell}$. The goal is to infer the group information based on the similarity matrix $A$.

The group information is contained in the expectation of the similarity matrix $A$. Specifically,
\[
\mathbb{E}A=\Theta\Psi\Theta^\top,
\]
where  $\Theta\in\{0,1\}^{n\times K}$ is the assignment matrix with $\Theta_{ik}=1$ if and only if $i\in S_{k}$. Notice that $\Theta$ spans the top $K$ eigenspace of $\mathbb{E}A$, and $\Theta_{i\cdot}=\Theta_{j\cdot}$ if and only if $i,j$ belong to the same group. Hence, the top $K$ eigenvectors $U^*\in\mathbb{R}^{n\times K}$ of $\mathbb{E}A$ satisfies that $U^*_{i\cdot}=U^*_{j\cdot}$ if and only if $i,j$ share the same group. Utilizing this observation, the spectral clustering algorithm~\cite{von2007tutorial} computes the top $K$ eigenvectors $U\in\mathbb{R}^{n\times K}$ of $A$ and then performs $k$-means clustering on the rows of $U$. It is worth noting that the $k$-means clustering results using the rows of $UO$ remain invariant no matter which orthogonal matrix $O\in\mathbb{R}^{K\times K}$ is used. Therefore, the efficacy of the spectral clustering algorithm depends on the subspace spanned by $U$, or equivalently, the associated projection matrix. 

Ideally, the projection matrix $X$ should also be associated with the assignment matrix $\Theta$ for the observed affinity $A$. The entries of such projection matrix $X$ are within $[0,\max_k\frac{1}{n_k}]$, where $n_k$ is the size of the $k$-th group. Therefore, instead of solving problem \eqref{equ:1}, we propose to solve the BPMA problem \eqref{BPA} with $\alpha=0$ and $\beta=\max_k\frac{1}{n_k}$. In practice, $\beta$ can be set as $K/n$ when the clusters are balanced, and it can be set as $\beta=\frac{cK}{n}$ with some constant $c>1$ when the clusters are unbalanced.



\section{Experiments}
We compare the proposed RBPMA approach with other community detection methods, including SDP-1~\cite{AMINI2018}, SDP-2~\cite{chen2016statistical}, and spectral clustering~\cite{von2007tutorial}, on both synthetic and real-world datasets. The details of these algorithms are listed in Table \ref{related algorithm}.
We use the solution of problem \eqref{equ:1} as the initial point $X_0$ for RBPMA and set the hyper-parameter $\lambda$ as a large enough constant such as $10^8$.
After solving $X$, we apply eigen-decomposition to compute its top $K$ eigenvectors $U\in\mathbb{R}^{n\times K}$. Then we normalize each row of $U$ to get $\tilde U\in\mathbb{R}^{n\times K}$. Finally, we perform $k$-means clustering on the rows of $\tilde U$. We evaluate the clustering results using two standard criteria:  accuracy (ACC) and normalized mutual information (NMI).

\subsection{Synthetic Data}
For synthetic data, we set $n=80$ and $K=3$. The sizes of the three groups are $30,20,30$ respectively. Let $\Psi\in\mathbb{R}^{3\times 3}$ be the connectivity probability matrix with $\Psi_{k\ell}=0.2,\,\forall k\neq \ell$, and $\Psi_{kk}= 0.49, \, \forall k \neq \ell$. In other four settings, we simply change $\Psi_{kk}$ to each of $\{0.46,0.43,0.40,0.37\}$. 
We generate the similarity matrix $A$ such that $A_{ij}\sim {\rm Bernoulli}(\Psi_{k\ell})$ if $i\in S_k$ and $j\in S_\ell$, where $S_k$ denotes the index set of the $k$-th group. 

\begin{table}[h]
\centering
\caption{The comparisons between four community detection algorithms using the synthetic data\label{ACCNMI}}
\resizebox{8.5cm}{!}{
\begin{tabular}{c|c|c|c|c||c|c|c|c}\hline \hline
& \multicolumn{4}{c||}{ACC} & \multicolumn{4}{c}{NMI} \\ \hline
$\Psi_{kk}$ &  SDP-1 &SDP-2 & Spectral & RBPMA  & SDP-1 &SDP-2 & Spectral & RBPMA  \\ \hline
 0.49& 0.974  & 0.712  & 0.948  & \bf 0.977  & 0.911  & 0.497  & 0.830  & \bf 0.916 \\
0.46 & 0.951  & 0.681  & 0.926  & \bf 0.955  & 0.835  & 0.405  & 0.767  &\bf  0.840 \\
0.43& 0.864  & 0.618  & 0.813  & \bf 0.873  & 0.626  & 0.308  & 0.534  & \bf 0.638 \\
0.40& 0.781  & 0.582  & 0.725  & \bf 0.786  & 0.478  & 0.237  & 0.386  & \bf 0.481 \\
0.37 & 0.666  & 0.512  & 0.650  & \bf 0.687  & 0.303  & 0.139  & 0.282  & \bf 0.305 \\ \hline \hline
 \end{tabular}}
 \end{table}

We apply all four community detection algorithms to $A$ and compute the corresponding ACC and NMI. In our RBPMA algorithm, we set $\alpha = 0$ and $\beta=1/20$. 
This experiment is repeated 20 times for each $\Psi$, and the average ACC and NMI are reported in Table~\ref{ACCNMI}.
The results show that our model outperforms its competitors uniformly. This agrees with the intuition that the quality of the rank-$K$ projection matrix is improved by using the extra boundedness constraint.




\subsection{Real Data}
\begin{table}[h]
\caption{The ACC and NMI  comparison for the four different methods on four datasets. \label{RealCompare}}
\resizebox{8.5cm}{!}{
\begin{tabular}{c|cccc||cccc}\hline\hline
& \multicolumn{4}{c||}{ACC} & \multicolumn{4}{c}{NMI} \\ \hline
Data &  SDP-1 &SDP-2 & Spectral & RBPMA  & SDP-1 &SDP-2 & Spectral & RBPMA  \\ \hline
 COIL-10       & 0.554  & 0.550  & 0.550  & \bf 0.570  & 0.606  & 0.603  & 0.605  &\bf  0.616 \\
 COIL-20      & 0.558  & 0.551  & 0.550  &\bf  0.564  & 0.610  & 0.606  & 0.604  &\bf  0.613 \\
  DIGIT-5 & 0.927  & 0.912  & 0.929  &\bf 0.936  & 0.834  & 0.826  & 0.830  & \bf 0.842 \\
 DIGIT-10 & 0.707  & 0.696  & 0.697  & \bf 0.716  & 0.659  & 0.655  & 0.653  &\bf 0.669 \\  \hline \hline
\end{tabular}}
 \end{table}


This subsection compares RBPMA with SDP-1, SDP-2, and spectral clustering on the Coil10, Coil20, and handwritten digit datasets. \cite{Coil}. Digit5 consists of 1000  images with the shape $15\times 16$ from five groups, and Digit 10 consists of 2000 images from 10 groups.
Since all groups of Coil and Digit datasets share the same size, we set $\alpha=0$ and $\beta = K/n$ in this experiment.
The similarity matrix $A$ is constructed using the Gaussian kernel: $A_{ij} = \exp(-\|x_i-x_j\|_2^2/\sigma^2)$, where $x_i$ is the $i$-th data vector and $\sigma^2=\frac{2}{n(n-1)}\sum_{i<j}\|x_i-x_j\|_2^2$ is the average of squared pairwise distances. Table \ref{RealCompare} shows that RBPMA outperforms its competitors in terms of ACC and NMI on all Coil and Digit datasets. 



%
%



\bibliographystyle{unsrt}
\bibliography{Cmm}

\begin{thebibliography}{10}

\bibitem{belkin2003laplacian}
Mikhail Belkin and Partha Niyogi.
\newblock Laplacian eigenmaps for dimensionality reduction and data
  representation.
\newblock {\em Neural computation}, 15(6):1373--1396, 2003.

\bibitem{cai2018comprehensive}
Hongyun Cai, Vincent~W Zheng, and Kevin Chen-Chuan Chang.
\newblock A comprehensive survey of graph embedding: Problems, techniques, and
  applications.
\newblock {\em IEEE transactions on knowledge and data engineering},
  30(9):1616--1637, 2018.

\bibitem{goyal2018graph}
Palash Goyal and Emilio Ferrara.
\newblock Graph embedding techniques, applications, and performance: A survey.
\newblock {\em Knowledge-Based Systems}, 151:78--94, 2018.

\bibitem{xu2021understanding}
Mengjia Xu.
\newblock Understanding graph embedding methods and their applications.
\newblock {\em SIAM Review}, 63(4):825--853, 2021.

\bibitem{hartigan1979k}
John~A Hartigan, Manchek~A Wong, et~al.
\newblock A k-means clustering algorithm.
\newblock {\em Applied statistics}, 28(1):100--108, 1979.

\bibitem{jung2014clustering}
Yong~Gyu Jung, Min~Soo Kang, and Jun Heo.
\newblock Clustering performance comparison using k-means and expectation
  maximization algorithms.
\newblock {\em Biotechnology \& Biotechnological Equipment}, 28(sup1):S44--S48,
  2014.

\bibitem{von2007tutorial}
Ulrike Von~Luxburg.
\newblock A tutorial on spectral clustering.
\newblock {\em Statistics and computing}, 17:395--416, 2007.

\bibitem{boyd2011distributed}
Stephen Boyd, Neal Parikh, Eric Chu, Borja Peleato, Jonathan Eckstein, et~al.
\newblock Distributed optimization and statistical learning via the alternating
  direction method of multipliers.
\newblock {\em Foundations and Trends{\textregistered} in Machine learning},
  3(1):1--122, 2011.

\bibitem{sun2022efficient}
Chuangchuang Sun.
\newblock An efficient approach for nonconvex semidefinite optimization via
  customized alternating direction method of multipliers.
\newblock {\em arXiv preprint arXiv:2209.03437}, 2022.

\bibitem{zhang2012low}
Zhenyue Zhang and Keke Zhao.
\newblock Low-rank matrix approximation with manifold regularization.
\newblock {\em IEEE transactions on pattern analysis and machine intelligence},
  35(7):1717--1729, 2012.

\bibitem{zhang2012inducible}
Zhenyue Zhang, Keke Zhao, and Hongyuan Zha.
\newblock Inducible regularization for low-rank matrix factorizations for
  collaborative filtering.
\newblock {\em Neurocomputing}, 97:52--62, 2012.

\bibitem{lee1999learning}
Daniel~D Lee and H~Sebastian Seung.
\newblock Learning the parts of objects by non-negative matrix factorization.
\newblock {\em Nature}, 401(6755):788--791, 1999.

\bibitem{ding2008convex}
Chris~HQ Ding, Tao Li, and Michael~I Jordan.
\newblock Convex and semi-nonnegative matrix factorizations.
\newblock {\em IEEE transactions on pattern analysis and machine intelligence},
  32(1):45--55, 2008.

\bibitem{kannan2012bounded}
Ramakrishnan Kannan, Mariya Ishteva, and Haesun Park.
\newblock Bounded matrix low rank approximation.
\newblock In {\em 2012 IEEE 12th International Conference on Data Mining}.

\bibitem{zhang2022graph}
Zhenyue Zhang, Zheng Zhai, and Limin Li.
\newblock Graph refinement via simultaneously low-rank and sparse
  approximation.
\newblock {\em SIAM Journal on Scientific Computing}, 44(3):A1525--A1553, 2022.

\bibitem{ji2011robust}
Hui Ji, Sibin Huang, Zuowei Shen, and Yuhong Xu.
\newblock Robust video restoration by joint sparse and low rank matrix
  approximation.
\newblock {\em SIAM Journal on Imaging Sciences}, 4(4):1122--1142, 2011.

\bibitem{richard2012estimation}
Emile Richard, Pierre-Andr{\'e} Savalle, and Nicolas Vayatis.
\newblock Estimation of simultaneously sparse and low rank matrices.
\newblock In {\em Proceedings of the 29th International Coference on
  International Conference on Machine Learning}, pages 51--58, 2012.

\bibitem{abbe2017community}
Emmanuel Abbe.
\newblock Community detection and stochastic block models: recent developments.
\newblock {\em The Journal of Machine Learning Research}, 18(1):6446--6531,
  2017.

\bibitem{AMINI2018}
Arash~A Amini and Elizaveta Levina.
\newblock On semidefinite relaxations for the block model.
\newblock {\em The Annals of Statistics}, 46(1):149--179, 2018.

\bibitem{chen2016statistical}
Yudong Chen and Jiaming Xu.
\newblock Statistical-computational tradeoffs in planted problems and submatrix
  localization with a growing number of clusters and submatrices.
\newblock {\em The Journal of Machine Learning Research}, 17(1):882--938, 2016.

\bibitem{Coil}
Dheeru Dua and Casey Graff.
\newblock {UCI} machine learning repository.
\newblock {\em University of California, Irvine, School of Information and
  Computer Sciences}, 2017.

\end{thebibliography}

\end{document}